\documentclass[final,3p]{elsarticle}
 \usepackage[latin1]{inputenc}
 \usepackage{graphics}
 \usepackage{graphicx}
 \usepackage{epsfig}
\usepackage{amssymb}
 \usepackage{amsthm}
 \usepackage{lineno}
 \usepackage{amsmath}
   \numberwithin{equation}{section}
\usepackage{mathrsfs}

\journal{ } %
\newtheorem{thm}{Theorem}[section]

\newtheorem{lem}[thm]{Lemma}

\newtheorem{defn}[thm]{Definition}

\newtheorem{exam}[thm]{Example}

\biboptions{sort&compress,square}
\allowdisplaybreaks
\begin{document}
\begin{frontmatter}
\author[rvt1]{Jian Wang}
\ead{wangj484@nenu.edu.cn}
\author[rvt2]{Yong Wang\corref{cor2}}
\ead{wangy581@nenu.edu.cn}
\cortext[cor2]{Corresponding author.}
\address[rvt1]{School of Science, Tianjin University of Technology and Education, Tianjin, 300222, P.R.China}
\address[rvt2]{School of Mathematics and Statistics, Northeast Normal University,
Changchun, 130024, P.R.China}

\title{The noncommutative residue, divergence theorems and \\the spectral geometry functional}
\begin{abstract}
In this paper, a simple proof of the divergence theorem is given by using the Dirac operator and non-commutative residues.
Then we extend the divergence theorem to compact manifolds with boundary by the   noncommutative residue of the $B$-algebra.
Furthermore, we present some spectral geometric functionals which we call spectral divergence functionals,  and we calculate the spectral divergence functionals of manifolds with (or without) boundary.
\end{abstract}
\end{frontmatter}

\section{Introduction}
\label{1}
In Connes program of noncommutative geometry,
for associative unital algebra $\mathcal{A}$ and Hilbert space $\mathcal{H}$ such that there is an algebra
 homomorphism $\pi:\mathcal{A}\rightarrow B(\mathcal{H})$, where $B(\mathcal{H})$
denotes the algebra of bounded operators acting on $\mathcal{H}$,
the role of geometrical
objects is played by spectral triples $(\mathcal{A}; \mathcal{H}; D )$.
Similar to the commutative case and
the canonical spectral triple $(C^{\infty}(M); L^{2}(S); D )$, where $(M; g; S )$ is a closed spin
manifold and $D$  is the Dirac operator acting on the spinor bundle $S$, the spectrum
of the Dirac operator $D$  of a spectral triple $(\mathcal{A}; \mathcal{H}; D )$ encodes the geometrical
information of the spectral triple.
An important new feature of such geometries, which is absent in the
commutative case, is the existence of inner fluctuations of the metric.
The spectral-theoretic approach to scalar curvature has been extended  to quantum tori in the seminal work of
 Connes and Tretkoff \cite{CoT}, expanded in \cite{CoM} and then extensively studied by many authors \cite{DaS,Fa,Ka,KW}.
No doubt it would be extremely interesting to recover scalar curvature, Ricci curvature and other important tensors in both the
 classical setup as well as for the noncommutative geometry or quantum geometry.
 Unlike the scalar curvature, the Ricci curvature does not appear in the coefficients
of the heat trace of the Dirac Laplacian. Floricel etc. \cite{FGK} observed that
the Laplacian of the de Rham complex, more precisely the Laplacian on one forms, captures the Ricci
operator in its second term, and formulated the Ricci
operator as a spectral functional on the algebra of sections of the endomorphism
bundle of the cotangent bundle of $M$:
\begin{align*}
Ric (F)=a_{2}(tr(F), \triangle_{0})-a_{2}(F, \triangle_{1}),~~~F\in C^{\infty}(\mathrm{End}(T^{*}M)).
\end{align*}
In  \cite{Co8}, Connes and Chamseddine proved in the general framework of noncommutative geometry that the inner
fluctuations of the metric coming from the Morita equivalence ${\mathcal{A}%
}\sim{\mathcal{A}}$ generate perturbations of $D$ of the form $D\rightarrow
D^{\prime}=D+A$, where the $A$ plays the role of the gauge potentials and is a
self-adjoint element of the bimodule
\begin{equation*}
\Omega_{D}^{1}=\,\Big\{\;\sum\,a_{j}\,[D,b_{j}]\;;\,a_{j},\,b_{j}\in{\mathcal{A}%
}\Big\}. \label{bim}%
\end{equation*}
For any of the $a_{j}$ for $j>0$ is a scalar, it defines a functional on the universal $n$-forms $\Omega^{n}%
({\mathcal{A}})$ \cite{Co8} by the equality
\begin{equation*}
\int_{\varphi}\,a_{0}\,da_{1}\cdots\,da_{n}=\;\varphi(a_{0},\,a_{1}%
,\cdots,a_{n}). \label{hoschtocyc}%
\end{equation*}
The following functional is then a Hochschild
cocycle \cite{Co8} and is given as Dixmier trace of infinitesimals of order one,
\begin{equation*}
\tau_{0}(a^{0},a^{1},a^{2},a^{3},a^{4})=\,{\int\!\!\!\!\!\!-}\,a^{0}%
\,[D,a^{1}]\,D^{-1}[D,a^{2}]\,D^{-1}[D,a^{3}]\,D^{-1}[D,a^{4}]\,D^{-1}.
\label{tau0}%
\end{equation*}
The Dixmier trace \cite{JD} of an operator $T \in \mathcal{L}^{1,\infty}
(\mathcal{H}) $
  measures the {\it logarithmic divergence} of
its ordinary trace, and the Dixmier trace is invariant under
  perturbations by trace class operators.

Let $\bar{v}$,$\bar{w}$ with the components with
respect to local coordinates $\bar{v}_{a}$ and $\bar{w}_{b}$, respectively, be two differential forms represented in
such a way as endomorphisms (matrices) $c(\bar{v}) $ and $c(\bar{w}) $ on the spinor bundle.
For $n = 2m$ dimensional spin manifold $M$, by the  operator $c(\bar{w})(Dc(\bar{v})+c(\bar{v})D)D^{-n+1}$ acting on sections of
 a vector bundle $S(TM)$ of rank $2^{m}$, Dabrowski etc. \cite{DL} obtained the Einstein
tensor (or, more precisely, its contravariant version)   from functionals
over the dual bimodule of one-forms:
  \begin{equation*}
Wres\big(c(\bar{w})(Dc(\bar{v})+c(\bar{v})D)D^{-n+1}\big)=\frac{\upsilon_{n-1}}{6}2^{m}\int_{M}(Ric^{ab}
-\frac{1}{2}Rg^{ab})\bar{v}_{a} \bar{w}_{b}{\rm vol}_{g},
\end{equation*}
where $g^{*}(\bar{v},\bar{w})=g^{ab}v_{a}w_{b}$ and $G(\bar{v},\bar{w})= (Ric^{ab}
-\frac{1}{2}sg^{ab})\bar{v}_{a} \bar{w}_{b}$ denotes the Einstein tensor evaluated on the two one-forms, where
$\bar{v}=\sum_{a=1}^{n}v_{a}dx^{a} $, $\bar{w}=\sum_{b=1}^{n}w_{b}dx^{b} $ and $\upsilon_{n-1}=\frac{2\pi^{m}}{\Gamma(m)}.$
  Dabrowski etc. \cite{DL} demonstrated that the noncommutative residue density recovered
the tensors $g$ and
\begin{equation*}
G:=Ric-\frac{1}{2}s(g)g,
\end{equation*}
as certain bilinear functionals of vector fields on a manifold $M$, while their
dual tensors are recovered as a density of bilinear functionals of differential one-forms on $M$.
Dirac operators with torsion are by now well-established analytical tools in the study of special geometric structures.
 Ackermann and  Tolksdorf \cite{AT} proved a generalized version of the well-known Lichnerowicz formula for the square of the
most general Dirac operator with torsion  $D_{T}$ on an even-dimensional spin manifold associated to a metric connection with torsion.
In \cite{PS,PS1}, Pf$\ddot{a}$ffle and Stephan considered orthogonal connections with arbitrary torsion on compact Riemannian manifolds,
and for the induced Dirac operators, twisted Dirac operators and Dirac operators of Chamseddine-Connes type they computed the spectral
action.
 Sitarz and Zajac \cite{SZ} investigated the spectral action for scalar perturbations of Dirac operators.
Iochum and Levy \cite{IL}computed the heat kernel coefficients for Dirac operators with one-form perturbations.
Wang \cite{Wa4} considered the arbitrary perturbations of Dirac operators, and established the associated Kastler-
Kalau-Walze theorem.
In \cite{WWW,WW2,WW3}, we gave two Lichnerowicz type formulas for Dirac operators
and signature operators twisted by a vector bundle with a non-unitary connection,  and proved the Kastler-Kalau-Walze type theorem associated to
Dirac operators with torsion on compact manifolds with boundary.
In \cite{LWW1}, Liu et al.  proved the Kastler-Kalau-Walze type theorems for the perturbation of the de Rham Hodge operator on 4-dimensional
and 6-dimensional compact manifolds. The trilinear functional of differential one-forms for a finitely summable regular spectral triple with a noncommutative residue has been computed by Dabrowski et al. in \cite{DSZ}, and for a canonical spectral triple over a closed spin manifold,
they recoverd the torsion of the linear connection, which is a first step towards linking the spectral approach with the algebraic approach based on Levi-Civita connections.
In \cite{WW4}, we provide an explicit computation of the spectral torsion
associated with the Connes type operator on even dimension compact manifolds, and we also
extend the spectral torsion for the Connes type operator to compact manifolds with boundary.
For a finitely summable regular spectral triple, we \cite{WW5} recover two forms, torsion of the linear connection and four forms by the noncommutative residue and perturbed de-Rham Hodge operators, and provide an explicit computation of generalized spectral forms associated with the perturbed de-Rham Hodge Dirac triple.

One of the important theorems in mathematics and physics is the divergence theorem, which links the surface integral of a vector function to its divergence volume integral over a closed surface.
More precisely, let us recall that in the   divergence theorem description  about closed surfaces.
Let $S$ be a positively-oriented closed surface with interior
$\Omega$, and let $\vec{F}$ be a vector field continuously differentiable in a domain contatining $\Omega$, then
\begin{equation*}
 \iiint_{\Omega}{\rm div} \vec{F}{\rm d}V=\iint_{S}\vec{F}{\rm d}S.
\end{equation*}
Let $(M, g)$ be a compact oriented Riemannian manifold, then, for every smooth
tangent vector field $X$ on $M$, a divergence theorem states that
\begin{equation*}
\int_{M}{\rm div}X {\rm d}v_{g}=0.
\end{equation*}
The starting point is the usual divergence theorem for the case where $X$ is smooth and $M$ is a Riemannian manifold,
then it would be interesting to recover the divergence theorem in both the
 classical setup as well as for the noncommutative geometry.
Motivated by  the usual divergence theorem  on compact oriented Riemannian manifold and  the spectral torsion \cite{ AT,PS,PS1,DSZ,WW4,WW5},
the purpose of this paper is to explore the divergence theorem  from the noncommutative residue density,
and give some new spectral geometry functionals which is the extension of the results in \cite{DSZ,WW2} to the  differential operator $[D^2,f]$. In addition, we extend the divergence theorem and the spectral divergence functionals  to compact manifolds with boundary.

\section{Preliminaries on the  symbols representation of Dirac operator}
In this section we fix notations and recall the previous work that will play a fundamental role here.
We also give a review on the  symbols representation of Dirac operator and the symbols
		of the higher inverse of the  differential operator $[D^2,f]$ and their relations.

 Let $M$ be a compact oriented spin manifold of even dimension $n=2m$ with Riemannian metric $g$, and
 let  $\mathcal{E}$ be the spinor bundle over $M$ with the Hermitian structure and the spin connection $\nabla^{\mathcal{E}}$.
\begin{defn}
A Dirac operator $D$ on a $\mathbb{Z}_{2}$-graded vector bundle $\mathcal{E}$ is a
first-order differential operator of odd parity on $\mathcal{E}$,
 \begin{align}
D:\Gamma(M,\mathcal{E}^{\pm})\rightarrow\Gamma(M,\mathcal{E}^{\mp}),
\end{align}
such that $D^{2}$ is a generalized Laplacian.
\end{defn}
 Let  $\nabla^L$ be the Levi-Civita connection about $g^M$, the Dirac operator $D$ is locally given as follows in terms of an orthonormal section $e_i$ (with dual section $\theta^k$) of the frame bundle of M \cite{Ka}:
\begin{align}
&D=i\gamma^i\widetilde{\nabla}_i=i\gamma^i(e_i+\sigma_i);\nonumber\\
&\sigma_i(x)=\frac{1}{4}\gamma_{ij,k}(x)\gamma^i\gamma^k=\frac{1}{8}\gamma_{ij,k}(x)[\gamma^j\gamma^k-\gamma^k\gamma^j],
\end{align}
where $\gamma_{ij,k}$ represents the Levi-Civita connection $\nabla$ with spin connection $\widetilde{\nabla}$, specifically:
\begin{align}
&\gamma_{ij,k}=-\gamma_{ik,j}=\frac{1}{2}[c_{ij,k}+c_{ki,j}+c_{kj,i}],~~~i,j,k=1,\cdot\cdot\cdot,4;\nonumber\\
&c_{ij}^k=\theta^k([e_i.e_j]).
\end{align}
Here the $\gamma^i$ are constant self-adjoint Dirac matrices s.t. $\gamma^i\gamma^j+\gamma^j\gamma^i=-2\delta^{ij}.$
 In the
 fixed orthonormal frame $\{e_1,\cdots,e_n\}$ and natural frames $\{\partial_1,\cdots,\partial_n\}$ of $TM$,
the connection matrix $(\omega_{s,t})$ defined by
\begin{equation}
\nabla^L(e_1,\cdots,e_n)= (e_1,\cdots,e_n)(\omega_{s,t}).
\end{equation}
 Then the Dirac operator has the form
\begin{equation}
D=\sum^n_{i=1}c(e_i)\nabla_{e_i}^{S}=\sum^n_{i=1}c(e_i)\Big[e_i
-\frac{1}{4}\sum_{s,t}\omega_{s,t}(e_i)c(e_s)c(e_t)\Big].
\end{equation}
Using the shorthand: $\Gamma^{k}=g^{ij}\Gamma_{ij}^{k};~~\sigma^{j}=g^{ij}\sigma_{i}$, we have respective symbols of $D^{2}$:
\begin{lem}\cite{Ka}\label{lemma1} For the generalized Laplacian $D^{2}$ , the symbols
		of the inverse of the Laplace operator read
\begin{align}
\sigma_2(D^2)&=|\xi|^2;\nonumber\\
\sigma_1(D^2)&=i(\Gamma^\mu-2\sigma^\mu)(x)\xi_\mu;\nonumber\\
\sigma_0(D^2)&=-g^{\mu\nu}(\partial^{x}_\mu\sigma_\nu+\sigma^\mu\sigma_\nu-\Gamma^\alpha_{\mu\nu}\sigma_\alpha)(x)+\frac{1}{4}s(x).
\end{align}
\end{lem}

For a pseudo-differential operator $P$, acting on sections of a spinor bundle over an even n-dimensional compact
Riemannian spin manifold $M$, the analogue of the volume element in noncommutative geometry is the operator $D^{-n}=:ds^n$. And pertinent operators are realized as pseudodifferential operators on the spaces of sections. Extending previous definitions by Connes \cite{Co2}, a
noncommutative integral was introduced in \cite{Fi} based on the noncommutative residue \cite{Wo,Wo1}, using the definition of the residue:
\begin{align}\label{666}
\int\hspace{-1.05em}- Pds^n:={\rm Wres}PD^{-n}:=\int_{S^*M} {\rm Tr}[\sigma_{-n}({PD^{-n}})](x,\xi),
\end{align}
where $\sigma_{-n}(PD^{-n})$ denotes the $-n$-th order piece of the complete symbols of $PD^{-n}$, {\rm tr} as shorthand of trace.

The following Lemmas play  a key role in the main symbol represents for Dirac operator  under the Lie
bracket $[D,f]$.
Let $S$ be pseudo-differential operator and $f$ is a smooth function, $[S,f]$ is a pseudo-differential operator,
the total symbol $\sigma([S,f][S,h])$ of the product of two pseudo-differential operators $[S,f]$ and $[S,h]$ is as follows.
\begin{lem}\cite{UW}\label{lema2}
Let $S$ be pseudo-differential operator of order $k$ and $f$ is a smooth function, $[S,f]$ is a pseudo-differential operator of order $k-1$ with total symbol
$\sigma[S,f]\sim\sum_{j\geq 1}\sigma_{k-j}[S,f]$, where
\begin{align}
\sigma_{k-j}[S,f]=\sum_{|\beta|=1}^j\frac{D_x^\beta(f)}{\beta!}\partial_\xi^\beta(\sigma^S_{k-(j-|\beta|)}).
\end{align}
\end{lem}

\begin{lem} \cite{UW}\label{lema2}
For $2k+n \geq 2,$ with the sum taken over $ |\alpha' |+|\alpha'' |+|\beta|+|\delta|+i+j=n+2k$,
$|\beta|\geq 1$, and $|\delta| \geq 1$,
\begin{align}
\sigma_{-n}([S,f][S,h])
=\sum \frac{D^\beta_x(f) D_x^{\alpha'+\delta}(h)}{\alpha'!\alpha'' !\beta!\delta!}
   \partial_\xi^{\alpha'+\alpha''+\beta}(\sigma^S_{k-i}) \partial_\xi^\delta(D_x^{\alpha``}(\sigma^S_{k-j})).
\end{align}
\end{lem}
\section{The divergence theorem for compact manifolds with boundary}

\subsection{The divergence theorem for compact manifolds without boundary }
Let $M$ be a compact oriented Riemannian manifold of even dimension $n=2m$, for Dirac operator and a smooth function $f$, we get
\begin{lem}
The noncommutative residue for the  differential operator $[D^2,f]$ is equal to
 \begin{equation}
Wres\big([D^2,f]D^{-2m}\big)=\int_{S^*M} {\rm Tr}[\sigma_{-2m}([D^2,f]D^{-2m})](x,\xi).
\end{equation}
\end{lem}

By Lemma 2.3, we get the following lemma.
\begin{lem}The symbols of $[D^2,f]$ are given
\begin{align}
&\sigma_0([D^2,f]) =\sum_{j=1}^n\partial_{x_j}(f)(\Gamma^j-2\sigma^j)(x)-\sum_{jl=1}^n\partial_{x_j}\partial_{x_l}(f)g^{jl};\\
&\sigma_1([D^2,f]) =-2\sqrt{-1}\sum_{jl=1}^n\partial_{x_j}(f)g^{jl}\xi_l.
\end{align}
\end{lem}
\begin{proof}
By (2.8) we have
\begin{align}
\sigma_1([D^2,f]) =&\sum_{j=1}^{2m}D_{x_{j}}(f)\partial_{\xi_{j}} (\sigma_{2}{D^{2}})
                  =-2\sqrt{-1}\sum_{jl=1}^{2m}\partial_{x_j}(f)g^{jl}\xi_l.
\end{align}
And
\begin{align}
\sigma_0([D^2,f])=&\sum_{j=1}^{2m}D_{x_{j}}(f)\partial_{\xi_{j}} (\sigma_{1}{D^{2}})
                   +\frac{1}{2}\sum_{j l=1}^{2m}D_{x_{j}}D_{x_{j}}(f)\partial_{\xi_{j}}\partial_{\xi_{l}} (\sigma_{2}{D^{2}})\nonumber\\
             &=\sum_{j=1}^{2m}\partial_{x_j}(f)(\Gamma^j-2\sigma^j)(x)-\sum_{jl=1}^{2m}\partial_{x_j}\partial_{x_l}(f)g^{jl}.
\end{align}
\end{proof}

By Lemma 3.4 in \cite{WUW} we have
\begin{lem}\cite{WUW}
General dimensional symbols about Dirac operator $D$ are given,
\begin{align}
\sigma_{-2m}(D^{-2m})&=|\xi|^{-2m}; \\
\sigma_{-2m-1}(D^{-2m})&=m|\xi|^{-(2m-2)}\bigg(-i|\xi|^{-4}\xi_k(\Gamma^k-2\sigma^k)-2i|\xi|^{-6}
\xi^j\xi_\alpha\xi_\beta\partial_{x_j}g^{\alpha\beta}\bigg)\nonumber\\
&+2i\sum_{k=0}^{m-2}\sum_{\mu=1}^{2m}(k+1-m)|\xi|^{-2m-4}\xi^\mu\xi_\alpha\xi_\beta\partial^x_\mu g^{\alpha\beta}.
\end{align}
 \end{lem}
Since $[D^2,f]$ is globally defined on $M$, so we can perform computations of $[D^2,f]$ in
normal coordinates. Taking normal coordinates about $x_0$, then
$\sigma^i(x_0)=0,~
\partial^j[c(\partial_j)](x_0)=0,~\Gamma^k(x_0)=0~g^{ij}(x_0)=\delta_i^j,~ \partial^x_\mu g^{\alpha\beta}(x_0)=0.  $
Based on the algorithm yielding the principal
symbol of a product of pseudo-differential operators in terms of the principal symbols of the factors, from  Lemma 3.1 and   Lemma 3.2,
the right side of equation (3.1) is equal to
\begin{align}
\sigma_{-2m}([D^2,f]D^{-2m})(x_0)
=&\left\{\sum_{|\alpha|=0}^\infty\frac{(-i)^{|\alpha|}}{\alpha!}\partial^\alpha_\xi[\sigma([D^2,f])]
\partial^\alpha_x[\sigma(D^{-2m})](x_0)\right\}_{-2m}\nonumber\\
=&\sigma_0([D^2,f])\sigma_{-2m}(D^{-2m})+\sigma_1([D^2,f])\sigma_{-2m-1}(D^{-2m})(x_0) \nonumber\\
& +(-i)\sum_{j=1}^{2m}\partial_{\xi_j}[\sigma_1([D^2,f])]\partial_{x_j}[\sigma_{-2m}(D^{-2m})](x_0)\nonumber\\
=&-\sum_{jl=1}^{2m}\partial_{x_j}\partial_{x_l}(f)g^{jl}(x_0).
\end{align}
Through symbolic operation, the left side of the equation (3.1) read
\begin{align}
Wres\big([D^2,f]D^{-2m}\big)=Wres\big((D^2f-f D^2)D^{-2m}\big)
=Wres\big(D^2f D^{-2m} \big)-Wres\big( f D^2D^{-2m}\big)=0.
\end{align}
Combining equations (3.8) and (3.9), we obtain the divergence theorem for compact manifolds
\begin{thm}Let $M$ be a compact oriented Riemannian manifold of even dimension $n=2m$ without boundary, the divergence theorem  for the generalized Laplacian $D^{2}$ and smooth function $f$ read
\begin{align}
Wres\big([D^2,f]D^{-2m}\big)=2^{m}\rm{vol}(S^{2m-1})\int_{ M}  \Delta(f)\rm{dvol}_{M}=0.
\end{align}
\end{thm}

\subsection{The divergence theorem for compact manifolds with boundary }
The purpose of this section is to specify the divergence theorem  on manifold with boundary.
Now we recall the main theorem in \cite{FGLS}.
\begin{thm}\cite{FGLS}
 Let $X$ and $\partial X$ be connected, ${\rm dim}X=n\geq3$,
 $A=\left(\begin{array}{lcr}\pi^+P+G &   K \\
T &  S    \end{array}\right)$ $\in \mathcal{B}$ , and denote by $p$, $b$ and $s$ the local symbols of $P,G$ and $S$ respectively.
 Define:
 \begin{align}
{\rm{\widetilde{Wres}}}(A)=&\int_X\int_{\bf S}{\mathrm{Tr}}_E\left[p_{-n}(x,\xi)\right]\sigma(\xi){\rm d}x \nonumber\\
&+2\pi\int_ {\partial X}\int_{\bf S'}\left\{{\mathrm{Tr}}_E\left[({\mathrm{Tr}}b_{-n})(x',\xi')\right]+{\mathrm{Tr}}
_F\left[s_{1-n}(x',\xi')\right]\right\}\sigma(\xi'){\rm d}x',
\end{align}
Then~~ a) ${\rm \widetilde{Wres}}([A,B])=0 $, for any
$A,B\in\mathcal{B}$;~~ b) It is a unique continuous trace on
$\mathcal{B}/\mathcal{B}^{-\infty}$.
\end{thm}

 Let $M$ be a compact oriented Riemannian manifold of even dimension $n=2m$ with boundary $\partial M$, we divide  (3.11) into two cases, according to the dimension $2m$.

 $\mathbf{Case ~~I}$:

 Denote by $\sigma_{l}(A)$ the $l$-order symbol of an operator A. An application of (2.1.4) in \cite{Wa1} shows that
\begin{equation}
\widetilde{{\rm Wres}}[\pi^+([D^2,f]D^{-1})\circ\pi^+D^{-2m+1}]
=\int_M\int_{|\xi|=1}{\rm Tr}_{S(TM)}[\sigma_{-2m}([D^2,f] D^{ -2m })]\sigma(\xi)\texttt{d}x+\int_{\partial
M}\Psi_{1},
\end{equation}
where $ r-k-|\alpha|+l-j-1=-2m,~~r\leq0,l\leq-2m+1$, and $r=0,~k=|\alpha|=j=0,l=-2m+1$, then
\begin{align}
\Psi_{1}=&\int_{|\xi'|=1}\int^{+\infty}_{-\infty}
{\rm Tr}_{S(TM)\otimes F}
[ \sigma^+_{0}( [D^2,f]D^{-1})(x',0,\xi',\xi_n)
 \partial_{\xi_n}\sigma_{-2m+1}(D^{-2m+1})(x',0,\xi',\xi_n)]\rm{d}\xi_{2m}\sigma(\xi')\rm{d}x'.
\end{align}
An easy calculation gives
  \begin{align}
& \pi^+_{\xi_n}\big( \sigma_{0}( [D^2,f]D^{-1})\big)\nonumber\\
=& \pi^+_{\xi_n}\big( \sigma_{1}( [D^2,f]) \sigma_{-1}(D^{-1})\big)\nonumber\\
                                  =& \pi^+_{\xi_n}\big( -2\sqrt{-1}\sum_{jl=1}^{2m}\partial_{x_j}(f)g^{jl}\xi_l \frac{\sqrt{-1}c(\xi)}{|\xi|^2}\big)\nonumber\\
                =& \pi^+_{\xi_n}\big( -2\sqrt{-1}\sum_{jl=1}^{2m-1}\partial_{x_j}(f)g^{jl}\xi_l \frac{\sqrt{-1}c(\xi)}{|\xi|^2}\big)
                  +\pi^+_{\xi_n}\big( -2\sqrt{-1} \partial_{x_n}(f)\xi_n \frac{\sqrt{-1}c(\xi)}{|\xi|^2}\big)\nonumber\\
     =& -2\sqrt{-1}\sum_{jl=1}^{2m-1}\partial_{x_j}(f)g^{jl}\xi_l  \frac{c(\xi')+\sqrt{-1}c(\rm{d}x_{n})}{2(\xi_n-i)}
                   -2\sqrt{-1} \partial_{x_n}(f)  \frac{\sqrt{-1}c(\xi')-c(\rm{d}x_{n})}{2(\xi_n-i)}.
 \end{align}
Where some basic facts and formulae about Boutet de Monvel's calculus which can be found  in Sec.2 in \cite{Wa1}.
By   Lemma 3.3, we get
  \begin{align}
\sigma_{-2m+1}(D^{-2m+1})=\sigma_{-2m+2}(D^{-2m+2})\sigma_{-1}(D^{-1})= \frac{\sqrt{-1}c(\xi)}{|\xi|^{2m}},
\end{align}
then
  \begin{align}
&\partial_{\xi_n}\big(\sigma_{-2m+1}(D^{-2m+1})\big)(x_0)
=\partial_{\xi_n}\big(\frac{\sqrt{-1}c(\xi)}{|\xi|^{2m}}\big)(x_0)\nonumber\\
=&\frac{-2\sqrt{-1}m\xi_nc(\xi')}{(1+\xi_n^2)^{m+1}}+\frac{ \sqrt{-1}(1+\xi_n^{2}-2m\xi_n^{2})c(dx_{n})}{(1+\xi_n^2)^{m+1}}.
\end{align}
By the relation of the Clifford action and $ {\rm{Tr}}(AB)= {\rm{Tr}}(BA) $, we get
\begin{align}
& {\rm Tr}_{S(TM)}[\sigma^+_{0}([D^2,f]D)(x',0,\xi',\xi_n)\partial_{\xi_n}\sigma_{-2m+1}(D^{-2m+1})(x',0,\xi',\xi_n)] \nonumber\\
 =& \frac{-2\sqrt{-1}m \partial_{x_n}(f) \xi_n}{(\xi_n-i)(1+\xi_n^2)^{m+1}}
 {\rm Tr}_{S(TM)}[ c(\xi')c(\xi')]\nonumber\\
 &+ \frac{ -2\sqrt{-1} \partial_{x_n}(f)(1+\xi_n^{2}-2m\xi_n^{2})}{(\xi_n-i)(1+\xi_n^2)^{m+1}} {\rm Tr}_{S(TM)} [c(dx_{n})c(dx_{n})].
 \end{align}
Also, straightforward computations yield
\begin{align}
\Big(\frac{i}{(\xi_n+i)^{m}}\Big)^{(m)}=&-im\big((\xi_n+i)^{-m-1}\big)^{(m-1)}
=  i(-1)^{2}m(m-1)\big((\xi_n+i)^{-m-2}\big)^{(m-2)}\nonumber\\
=& \cdots
=(-1)^{m}\frac{(2m-1)!}{(m-1)!}(2i)^{-2m},
\end{align}
and
\begin{align}
\Big(\frac{\xi_n}{(\xi_n+i)^{m}}\Big)^{(m)}|_{\xi_n=i}
=&\Big[\Big(\frac{1}{(\xi_n+i)^{m-1}}\Big)^{(m)}-\Big(\frac{i}{(\xi_n+i)^{m}}\Big)^{(m)}\Big]|_{\xi_n=i}
=\frac{(2m-2)!(-i)2^{-2m}}{(m-1)!}.
\end{align}
Substituting into above calculation we get
\begin{align}
\Psi_{1}=&\int_{|\xi'|=1}\int^{+\infty}_{-\infty}
{\rm Tr}_{S(TM)\otimes F}
\Big[ \sigma^+_{0}([D^2,f]D)(x',0,\xi',\xi_n)\nonumber\\
&\times
\partial_{\xi_n}\sigma_{-2m+2}(D^{-2m+1})(x',0,\xi',\xi_n)\Big]\rm{d}\xi_{2m}\sigma(\xi')\rm{ d}x'\nonumber\\
=&\int_{|\xi'|=1}\int^{+\infty}_{-\infty}\Big(
 \frac{-2\sqrt{-1}m \partial_{x_n}(f) \xi_n}{(\xi_n-i)(1+\xi_n^2)^{m+1}}
 {\rm Tr}_{S(TM)}[ c(\xi')c(\xi')] \nonumber\\
 &- \frac{   \partial_{x_n}(f)(1+\xi_n^{2}-2m\xi_n^{2})}{(\xi_n-i)(1+\xi_n^2)^{m+1}} {\rm Tr}_{S(TM)} [c(dx_{n})c(dx_{n})] \Big)
 \rm{d}\xi_{2m}\sigma(\xi')\rm{d}x'\nonumber\\
 =& \frac{(2m-1)! 2^{-2m+1}\sqrt{-1}\pi}{(m-1)!m!} \partial_{x_n}(f) {\rm Tr}_{S(TM)} [ {\rm Id}] {\rm{vol}}(S^{2m-2})
{\rm d}\rm{vol}_{\partial_{M}}.
\end{align}
Summing up (3.10) and (3.20) leads to the divergence theorem for manifold with boundary  as follows.
\begin{thm} Let $M$ be an $n=2m$ dimension compact oriented Riemannian manifold with boundary $\partial M$, the divergence theorem for the generalized Laplacian $D^{2}$ and smooth function $f$ read
\begin{align}
\widetilde{{\rm Wres}}[\pi^+([D^2,f]D^{-1})\circ\pi^+D^{-2m+1}]
=&2^{m}\rm{vol}(S^{2m-1})\int_{ M}  \Delta(f)\rm{dvol}_{M}\nonumber\\
 &+\int_{\partial_{M}}\frac{(2m-1)! 2^{-2m+1}\sqrt{-1}\pi}{(m-1)!m!} \partial_{x_n}(f) {\rm Tr}  [ {\rm Id}] {\rm{vol}}(S^{2m-2})
{\rm d}\rm{vol}_{\partial_{M}}.
\end{align}
\end{thm}
Remark: Our divergence theorem is different with the usual divergence theorem for manifolds with boundary.

$\mathbf{Case ~~II}$:
\begin{equation}
\widetilde{{\rm Wres}}[\pi^+([D^2,f]D^{-2})\circ\pi^+D^{-2m+2 }]
=\int_M\int_{|\xi|=1}{\rm Tr}_{S(TM)}[\sigma_{-n}([D^2,f] D^{ -2m+2})]\sigma(\xi)\texttt{d}x+\int_{\partial
M}\Psi_{2},
\end{equation}
where $ r-k-|\alpha|+l-j-1=-2m,~~r\leq -1,l\leq-2m+2$, and $r=-1,~k=|\alpha|=j=0,l=-2m+2$, then
\begin{align}
\Psi_{2}=&\int_{|\xi'|=1}\int^{+\infty}_{-\infty}
{\rm Tr}_{S(TM)\otimes F}
[ \sigma^+_{-1}( [D^2,f]D^{-2})(x',0,\xi',\xi_n)
 \partial_{\xi_n}\sigma_{-2m+2}(D^{-2m+2 })(x',0,\xi',\xi_n)]\rm{d}\xi_{2m}\sigma(\xi')\rm{d}x'.
\end{align}
An easy calculation gives
  \begin{align}
& \pi^+_{\xi_n}\big( \sigma_{-1}( [D^2,f]D^{-2})\big)\nonumber\\
=& \pi^+_{\xi_n}\big( \sigma_{1}( [D^2,f]) \sigma_{-2}(D^{-2})\big)\nonumber\\
                                  =& \pi^+_{\xi_n}\big( -2\sqrt{-1}\sum_{jl=1}^{2m}\partial_{x_j}(f)g^{jl}\xi_l \frac{1}{|\xi|^2}\big)\nonumber\\
                =& \pi^+_{\xi_n}\big( -2\sqrt{-1}\sum_{jl=1}^{2m-1}\partial_{x_j}(f)g^{jl}\xi_l \frac{1}{|\xi|^2}\big)
                  +\pi^+_{\xi_n}\big( -2\sqrt{-1} \partial_{x_n}(f)\xi_n \frac{1}{|\xi|^2}\big)\nonumber\\
     =& -2\sqrt{-1}\sum_{jl=1}^{2m-1}\partial_{x_j}(f)g^{jl}\xi_l  \frac{-\sqrt{-1}}{2(\xi_n-i)}
                   -2\sqrt{-1} \partial_{x_n}(f)  \frac{1}{2(\xi_n-i)}.
 \end{align}
And
   \begin{align}
 \partial_{\xi_n}\big(\sigma_{-2m+2}(D^{-2m+2})\big)(x_0)
=\partial_{\xi_n}\big(|\xi|^{-2m+2}\big) (x_0)
= \frac{2(1-m)\xi_n }{(1+\xi_n^2)^{m }}.
\end{align}
By the relation of the Clifford action and $ {\rm{Tr}}(AB)= {\rm{Tr}}(BA) $, we get
\begin{align}
& {\rm Tr}_{S(TM)}[\sigma^+_{-1}([D^2,f]D^{-2})(x',0,\xi',\xi_n)\partial_{\xi_n}\sigma_{-2m+2}(D^{-2m+2})(x',0,\xi',\xi_n)] \nonumber\\
 =& \frac{2(m-1)  \partial_{x_j} \sum _{j=1}^{2m-1}\partial_{x_j}(f) \xi_n}{(\xi_n-i)(1+\xi_n^2)^{m }}
 {\rm Tr}_{S(TM)} [ {\rm Id}]\nonumber\\
 &+ \frac{ -2\sqrt{-1} \partial_{x_n}(f)(1-m)\xi_n}{(\xi_n-i)(1+\xi_n^2)^{m }} {\rm Tr}_{S(TM)} [ {\rm Id}].
 \end{align}
Substituting into above calculation we get
\begin{align}
\Psi_{2}=&\int_{|\xi'|=1}\int^{+\infty}_{-\infty}
{\rm Tr}_{S(TM)\otimes F}
[ \sigma^+_{-1}([D^2,f]D^{2})(x',0,\xi',\xi_n)\nonumber\\
&\times
\partial_{\xi_n}\sigma_{-2m+2}(D^{-2m+2})(x',0,\xi',\xi_n)]\rm{d}\xi_{2m}\sigma(\xi')\rm{ d}x'\nonumber\\
=&\int_{|\xi'|=1}\int^{+\infty}_{-\infty}\Big(
\frac{2(m-1)  \partial_{x_j} \sum _{j=1}^{2m-1}\partial_{x_j}(f) \xi_n}{(\xi_n-i)(1+\xi_n^2)^{m }}
 {\rm Tr}_{S(TM)} [ {\rm Id}]\nonumber\\
  &+ \frac{ -2\sqrt{-1} \partial_{x_n}(f)(1-m)\xi_n}{(\xi_n-i)(1+\xi_n^2)^{m }} {\rm Tr}_{S(TM)} [ {\rm Id}]. \Big)
 \rm{d}\xi_{2m}\sigma(\xi')\rm{d}x'\nonumber\\
 =& \frac{(2m-2)! 2^{-2m+2}\sqrt{-1}\pi}{(m-2)!m!} \partial_{x_n}(f) {\rm Tr}_{S(TM)} [ {\rm Id}] {\rm{vol}}(S^{2m-2})
{\rm d}\rm{vol}_{\partial_{M}}.
\end{align}
Summing up (3.10) and (3.27) leads to the divergence theorem for manifold with boundary  as follows.
\begin{thm} Let $M$ be an $n=2m$ dimension compact oriented Riemannian manifold with boundary $\partial M$, the divergence theorem for the generalized Laplacian $D^{2}$ and smooth function $f$ read
\begin{align}
\widetilde{{\rm Wres}}[\pi^+([D^2,f]D^{-2})\circ\pi^+D^{-2m+2}]
=&2^{m}\rm{vol}(S^{2m-1})\int_{ M}  \Delta(f)\rm{dvol}_{M}\nonumber\\
 &+ \int_{\partial_{M}}\frac{(2m-2)! 2^{-2m+2}\sqrt{-1}\pi}{(m-2)!m!} \partial_{x_n}(f) {\rm Tr}  [ {\rm Id}] {\rm{vol}}(S^{2m-2})
{\rm d}\rm{vol}_{\partial_{M}}.
\end{align}
\end{thm}

\section{The spectral divergence functionals  }

 In this section we consider an $n$-dimensional oriented Riemannian manifold $(M, g^{M})$ equipped
with some spin structure. The Levi-Civita connection
$\nabla: \Gamma(TM)\rightarrow \Gamma(T^{*}M\otimes TM)$ on $M$ induces a connection
$\nabla^{S}: \Gamma(S)\rightarrow \Gamma(T^{*}M\otimes S).$
 Then we have the following definition.
\begin{defn}
The  perturbed Dirac operators   is the first order differential operator on $ \Gamma(M,S(TM) )$  given by the formula
 \begin{align}
\widetilde{D} = \sum_{i}^{n}c(e_{i})\nabla^{ S(TM) }_{e_{i}}+  \Phi ,
\end{align}
where $\Phi=\sum _{1\leq j_{1}\leq \cdots\leq j_{l}\leq n} \Phi(e_{j_{1}}, \cdots,e_{j_{l}})c(e_{j_{1}})\cdots c(e_{j_{l}}) $   denotes the Clifford multiplication by any form.

\end{defn}
It is natural to study the properties of Dirac operators with different structural torsions.
\begin{exam}
Suppose that  $\Phi=\gamma$, the grading operator on $S(TM)$, then
 $D +\gamma $  is the  Connes   type operator.
\end{exam}

\begin{exam} Let $\Phi=c(X)\gamma$ be the Clifford multiplication by the vector field $X$,
then $D +c(X)\gamma $  is the
twisted  type operator  in \cite{MNZ}.
\end{exam}
\begin{exam} Let $\Phi=\sqrt{-1}c(T)\gamma$ be the Clifford multiplication by the 3-form $T$ and the grading operator $\gamma$ on $S(TM)$,
then $D +\sqrt{-1}c(T)\gamma $  is the twisted  type operator with torsion.
\end{exam}
\begin{exam}
Assume  $\Phi=(c(T)+\sqrt{-1}c(Y))$ with  $c(T)=\sum _{1\leq\alpha<\beta<\gamma\leq n}
 T_{\alpha \beta \gamma}c(e_{\alpha})c(e_{\beta})c(e_{\gamma})$ and the Clifford multiplication by the vector field $Y$ is skew-adjoint, then
$D +(c(T)+\sqrt{-1}c(Y))  $  is the general twisted Dirac operator with torsion.
\end{exam}

\subsection{The  spectral divergence functionals   for Dirac operators with torsion  }
Let $u$,$v$ with the components with
respect to local coordinates $u_{a}$ and $v_{b}$, respectively, be two vector fields represented in
such a way as endomorphisms (matrices) $c(u) $ and $c(v) $ on the spinor bundle.
For $n = 2m$ dimensional compact oriented Riemannian manifold $(M,g)$, motivated by Theorem 3.4,  we introduce the  spectral divergence functionals
$Wres\big(c(u)[D^2,f]c(v)D^{-2m}\big) $, which is easily extend to the general spectral triple.
 For the Dirac operator with torsion $\widetilde{D} = D+  \Phi ,$ and a smooth function $f$,  the   spectral divergence functional is defined as follows.
\begin{lem}
The   spectral divergence functionals for the differential operator $[\widetilde{D}^2,f]$  multiplied by $c(u) $,$c(v) $ is equal to
 \begin{equation}
Wres\big(c(u)[\widetilde{D}^2,f]c(v)\widetilde{D}^{-2m}\big)=\int_{S^*M} {\rm Tr}[\sigma_{-2m}(c(u)[\widetilde{D}^2,f]c(v)\widetilde{D}^{-2m})](x,\xi).
\end{equation}
\end{lem}
By Lemma 2.4 and Lemma 3.2, we get the following   main symbol represents.
\begin{lem}The symbols of $[\widetilde{D}^2,f]$ are given
\begin{align}
&\sigma_1(\widetilde{D}^2)= \sqrt{-1}(\Gamma^j-2\sigma^j)\xi_{j} -\sum_{i=1}^{2m}\big(c(\partial_{i})\Phi +\Phi c(\partial_{i})\big)\partial_{i};\\
&\sigma_2(\widetilde{D}^2)= |\xi|^{2} ;\\
&\sigma_0([\widetilde{D}^2,f])=\sum_{j=1}^n\partial_{x_j}(f)\Big(\Gamma^j-2\sigma^j+\sum_{i=1}^{2m}\big(c(\partial_{i})\Phi +\Phi c(\partial_{i})\big)\Big) -\sum_{jl=1}^n\partial_{x_j}\partial_{x_l}(f)g^{jl};\\
&\sigma_1([\widetilde{D}^2,f])=-2\sqrt{-1}\sum_{jl=1}^n\partial_{x_j}(f)g^{jl}\xi_l.
\end{align}
\end{lem}
By Lemma 2.4, we have
\begin{align}
\sigma_1([\widetilde{D}^2,f]c(v))=& \sigma_1([\widetilde{D}^2,f])\sigma_0(c(v))+\sigma_0([\widetilde{D}^2,f])\sigma_1(c(v))\nonumber\\
  &-\sqrt{-1}\sum_{j=1}^{n}\partial_{\xi_{j}}\big(\sigma_1([\widetilde{D}^2,f])\big)\partial _{x_{j}}\big(\sigma_0\big(c(v)\big),
\end{align}
and
\begin{align}
\sigma_0([\widetilde{D}^2,f]c(v))= \sigma_0([\widetilde{D}^2,f])\sigma_0(c(v))
  -\sqrt{-1}\sum_{j=1}^{n}\partial\xi_{j}\big(\sigma_1([\widetilde{D}^2,f])\big)\partial x_{j}\big(\sigma_0\big(c(v)\big).
\end{align}
From (4.10), (4.11) and Lemma 4.7, we obtain
\begin{lem}The symbols of $[\widetilde{D}^2,f]c(v)$ are given
\begin{align}
\sigma_0([\widetilde{D}^2,f]c(v))=&\sum_{j=1}^n\partial_{x_j}(f)\Big(\Gamma^j-2\sigma^j+\sum_{i=1}^{2m}\big(c(\partial_{i})\Phi +\Phi c(\partial_{i})\big)\Big)c(v) \nonumber\\
     &-\sum_{jl=1}^n\partial_{x_j}\partial_{x_l}(f)g^{jl}c(v)-2\sum_{j=1}^n\partial_{x_j}(f)\partial_{x_j}(c(v)) ;\\
\sigma_1([\widetilde{D}^2,f]c(v))=&-2\sqrt{-1}\sum_{jl=1}^n\partial_{x_j}(f)g^{jl}\xi_lc(v).
\end{align}
\end{lem}
 Taking normal coordinates about $x_0$, then
$\sigma^i(x_0)=0,~ \Gamma^k(x_0)=0~g^{ij}(x_0)=\delta_i^j,~ \partial^x_\mu g^{\alpha\beta}(x_0)=0.  $ By Lemma 3.3   we obtain
\begin{lem}
General dimensional symbols about Dirac operator $D$ are given,
\begin{align}
\sigma_{-2m}(\widetilde{D}^{-2m})(x_0)&=|\xi|^{-2m}; \\
\sigma_{-2m+1}(\widetilde{D}^{-2m})(x_0)&=-\sqrt{-1}m|\xi|^{-2m-2}\sum_{i=1}^{2m}\big(c(\partial_{i})\Phi +\Phi c(\partial_{i})\big)\xi_{j}.
\end{align}

\end{lem}

Based on the algorithm yielding the principal
symbol of a product of pseudo-differential operators in terms of the principal symbols of the factors, from  Lemma 4.8 and   Lemma 4.9,
 we obtain
\begin{align}
&\sigma_{-2m}(c(u)[\widetilde{D}^2,f]c(v)\widetilde{D}^{-2m})(x_0)\nonumber\\
=&\left\{\sum_{|\alpha|=0}^\infty\frac{(-i)^{|\alpha|}}{\alpha!}c(u)\partial^\alpha_\xi[\sigma([\widetilde{D}^2,f])c(v)]
\partial^\alpha_x[\sigma(\widetilde{D}^{-2m})](x_0)\right\}_{-2m}\nonumber\\
=&c(u)\sigma_0([D^2,f]c(v))\sigma_{-2m}(\widetilde{D}^{-2m})+c(u)\sigma_1([\widetilde{D}^2,f]c(v))\sigma_{-2m-1}(\widetilde{D}^{-2m})(x_0) \nonumber\\
& +(-i)\sum_{j=1}^{2m}c(u)\partial_{\xi_j}[\sigma_1([\widetilde{D}^2,f])c(v)]\partial_{x_j}[\sigma_{-2m}(\widetilde{D}^{-2m})](x_0)\nonumber\\
=&H_{1}(x_0)+H_{2}(x_0)+H_{3}(x_0)+H_{4}(x_0),
\end{align}
where
\begin{align}
H_{1}(x_0)=&\sum_{j=1}^n\partial_{x_j}(f)c(u) \sum_{i=1}^{2m}\big(c(\partial_{i})\Phi +\Phi c(\partial_{i})\big)c(v) |\xi|^{-2m};\nonumber\\
  H_{2}(x_0)=& -\sum_{jl=1}^n\partial_{x_j}\partial_{x_l}(f)g^{jl}c(u)c(v)|\xi|^{-2m};\nonumber\\
H_{3}(x_0)=& -2 \sum_{j=1}^n\partial_{x_j}(f) c(u)  \partial_{x_j}(c(v))|\xi|^{-2m};\nonumber\\
   H_{4}(x_0)=&-2m \sum_{j=1}^n\partial_{x_j}(f) g^{jl} \xi_lc(u)c(v)\sum_{i=1}^{2m}\big(c(\partial_{i})\Phi +\Phi c(\partial_{i})\big)\xi_j  |\xi|^{-2m}.\nonumber
\end{align}
Substituting (4.16) into (4.5) yields the spectral divergence theorem for Dirac operators with torsion.

\subsection{The spectral divergence theorem  for the Dirac operators with twisted fluctuation}

The grading operator $\gamma$ denote by
\begin{align}
\gamma= (\sqrt{-1})^{m}\prod_{j=1}^{2m}c(e_{j}).
\end{align}
In the terms of the orthonormal frames $\{e_{i}\}(1\leq i,j\leq n)$ on $TM$, we have
 $
\gamma=(\sqrt{-1})^{m}c(e_{1})c(e_{2})\cdots c(e_{n}).
$
\begin{defn}\cite{Y}
 Suppose $V$ is a super vector space, and $\gamma$ is its super structure. If $\phi\in \rm{End}(V)$,
 let ${\rm{Tr}}(\phi)$ be the trace of $\phi$, then define
 \begin{align}
{\rm{ Str}}(\phi)={\rm{Tr}}(\gamma\circ\phi).
\end{align}
 Here ${\rm{Tr}}$ and $Str$ are called the trace and the super trace of $\phi$ respectively.
 \end{defn}
 \begin{lem}\cite{Y}
The super trace (function) $Str:\rm{End}_{\mathbb{C}}(S(2m))\rightarrow \mathbb{C}$
is a complex linear map satisfying
\begin{eqnarray}
{\rm{ Str}}(c(e_{i_{1}})c(e_{i_{w}})\cdots c(e_{i_{q}})) &=&\left\{
       \begin{array}{c}
        0,  ~~~~~~~~~~~~~~~~~~~~~~~{\rm if }~q<2m; \\[2pt]
      \frac{2^{m}}{(\sqrt{-1})^{m}}, ~~~~~~~~~~~~~~~{\rm if }~q=2m,
       \end{array}
    \right.
\end{eqnarray}
where $1\leq i_{1},i_{2},i_{q}\leq 2m $.
\end{lem}
Now we explore the spectral  divergence theorem  for the Dirac operators with twisted fluctuation.

{\bf Case (1)} The  spectral  divergence theorem for $D+\gamma $

 Substituting $D +\gamma  $ into (4.16) yields
 \begin{align}
\sigma_{-2m}(c(u)[(D+\gamma)^2,f]c(v)(D+\gamma)^{-2m})(x_0)
=\Big(H_{1}(x_0)+H_{2}(x_0)+H_{3}(x_0)+H_{4}(x_0)\Big)\Big|_{\Phi=\gamma},
\end{align}
where
\begin{align}
H_{1}(x_0)\Big|_{\Phi=\gamma}=&\sum_{j=1}^n\partial_{x_j}(f)c(u) \sum_{i=1}^{2m}\big(c(\partial_{i})\gamma +\gamma c(\partial_{i})\big)c(v) |\xi|^{-2m};\nonumber\\
  H_{2}(x_0)\Big|_{\Phi=\gamma}=& -\sum_{jl=1}^n\partial_{x_j}\partial_{x_l}(f)g^{jl}c(u)c(v)|\xi|^{-2m};\nonumber\\
H_{3}(x_0)\Big|_{\Phi=\gamma}=& -2 \sum_{j=1}^n\partial_{x_j}(f) c(u)  \partial_{x_j}(c(v))|\xi|^{-2m};\nonumber\\
   H_{4}(x_0)\Big|_{\Phi=\gamma}=&-2m \sum_{j=1}^n\partial_{x_j}(f) g^{jl} \xi_lc(u)c(v)\sum_{i=1}^{2m}\big(c(\partial_{i})\gamma
   +\gamma c(\partial_{i})\big)\xi_j  |\xi|^{-2m}.\nonumber
\end{align}
By the relation of the Clifford action and Lemma 4.1, we get
\begin{align}
{\rm{Tr}}\Big(c(u) \sum_{i=1}^{2m}\big(c(\partial_{i})\gamma +\gamma c(\partial_{i})\big)c(v) \Big)
=& \sum_{i=1}^{2m}{\rm{Tr}}\Big(c(u) c(\partial_{i})\gamma  c(v) )\Big)
 +{\rm{Tr}}\Big(c(u) \sum_{i=1}^{2m}\gamma c(\partial_{i}) c(v) )\Big) \nonumber\\
=&\sum_{i=1}^{2m}{\rm{Tr}}\Big(c(v)c(u) c(\partial_{i})\gamma  )\Big)
 +\sum_{i=1}^{2m}{\rm{Tr}}\Big(\gamma c(\partial_{i}) c(v)c(u) )\Big) =0.
\end{align}

\begin{align}
 {\rm{Tr}}\Big(-\sum_{jl=1}^n\partial_{x_j}\partial_{x_l}(f)g^{jl}c(u)c(v)\Big)
=  -\sum_{jl=1}^n\partial_{x_j}\partial_{x_l}(f)g^{jl}{\rm{Tr}}\Big(c(u)c(v)\Big)
= -\Delta(f)g(u,v){\rm Tr}  [ {\rm Id}].
\end{align}

\begin{align}
 {\rm{Tr}}\Big(-2 \sum_{j=1}^n\partial_{x_j}(f) c(u)  \partial_{x_j}(c(v))\Big)
=& -2 \sum_{j=1}^n\partial_{x_j}(f){\rm{Tr}}\big(c(u) \partial_{x_j}(c(v))\big)
= -2 {\rm Tr}  [c(u) {\rm grad}f(c(v))] \nonumber\\
=&  -2 {\rm Tr}  [c(u) \nabla^{S(TM)}_{{\rm grad}f}(c(v))]
=  -2 {\rm Tr}  [c(u) c(\nabla^{S(TM)}_{{\rm grad}f}v)] \nonumber\\
=&  2g(u,\nabla^{S(TM)}_{{\rm grad}f}v){\rm Tr}  [ {\rm Id}].
\end{align}

Then we obtain
\begin{thm}
The spectral noncommutative residue for the  differential operator $[(D+\gamma)^2,f]$  multiplied by $c(u) $,$c(v) $ is equal to
 \begin{equation}
Wres\big(c(u)[(D+\gamma)^2,f]c(v)(D+\gamma)^{-2m}\big)=\int_{M}  \Big(-\Delta(f)g(u,v)+ 2g(u,\nabla^{S(TM)}_{{\rm grad}f}v)  \Big){\rm Tr}_{S(TM)}[ {\rm Id}] {\rm dvol}_{ M} .
\end{equation}
\end{thm}

{\bf Case (2)} The spectral  divergence theorem for  $D + c(X)\gamma $

 Let $X=\sum_{i=1}^{n}X_{i}e_{i}$,  substituting $D + c(X)\gamma $ into (4.16) yields
 \begin{align}
&\sigma_{-2m}(c(u)[(D+c(X)\gamma)^2,f]c(v)(D+c(X)\gamma)^{-2m})(x_0)  \nonumber\\
=&\Big(H_{1}(x_0)+H_{2}(x_0)+H_{3}(x_0)+H_{4}(x_0)\Big)\Big|_{\Phi=c(X)\gamma},
\end{align}
where
\begin{align}
H_{1}(x_0)\Big|_{\Phi=c(X)\gamma}=&\sum_{j=1}^n\partial_{x_j}(f)c(u) \sum_{i=1}^{2m}\big(c(\partial_{i})c(X)\gamma +c(X)\gamma c(\partial_{i})\big)c(v) |\xi|^{-2m};\nonumber\\
  H_{2}(x_0)\Big|_{\Phi=c(X)\gamma}=& -\sum_{jl=1}^n\partial_{x_j}\partial_{x_l}(f)g^{jl}c(u)c(v)|\xi|^{-2m};\nonumber\\
H_{3}(x_0)\Big|_{\Phi=c(X)\gamma}=& -2 \sum_{j=1}^n\partial_{x_j}(f) c(u)  \partial_{x_j}(c(v))|\xi|^{-2m};\nonumber\\
   H_{4}(x_0)\Big|_{\Phi=c(X)\gamma}=&-2m \sum_{j=1}^n\partial_{x_j}(f) g^{jl} \xi_lc(u)c(v)\sum_{i=1}^{2m}\big(c(\partial_{i})c(X)\gamma
   +c(X)\gamma c(\partial_{i})\big)\xi_j  |\xi|^{-2m}.\nonumber
\end{align}
When $n=2m=4$, by the relation of the Clifford action and Theorem 4.11 in \cite{WW4}, we get
\begin{align}
{\rm{Tr}}\big(\sum_{j=1}^n \partial_{x_j}(f)c(u)c(v)c(\partial_{j})c(X)\gamma )\big)
 =&\sum_{j,k,l,i=1}^{n}  u_{j}v_{k}\partial_{x_j}(f)X_{i}{\rm{Tr}}\big(c(e_{j})c(e_{k}) c(e_{l}) c(e_{i})
      \gamma\big)\nonumber\\
 =&\sum_{i,j,k,l=1}^{n}  u_{j}v_{k}\partial_{x_j}(f)X_{i}{\rm{ Str}}\big(c(e_{j})c(e_{k}) c(e_{l}) c(e_{i})\big)\nonumber\\
  =&\sum_{\{i,j,k,l \}=\{1,2,3,4\}}   \frac{2^{2}}{(\sqrt{-1})^{2}}u_{j}v_{k}\partial_{x_j}(f)X_{i}\nonumber\\
   =&-4\langle  u^{*}\wedge v^{*}\wedge  ({\rm grad}f)^{*}\wedge X^{*},e_{1}^{*}\wedge e_{2}^{*}\wedge e_{3}^{*}\wedge e_{4}^{*}  \rangle.
\end{align}
 Then we obtain
\begin{equation}
{\rm{Tr}}\big(\sum_{j=1}^nc(u)c(v)c(\partial_{x_j}(f))c(X)\gamma \big)=\left\{
       \begin{array}{c}
        0,  ~~~~~~~~~~~~~~~~~~~~~~~~~~~~~~~~~~~~  ~~~~~~~~~~~~~~~~~~~~~~~ ~~{\rm if }~2m\neq4; \\[2pt]
      -4\langle u^{*}\wedge v^{*}\wedge ({\rm grad}f)^{*}\wedge X^{*},e_{1}^{*}\wedge e_{2}^{*}\wedge e_{3}^{*}\wedge e_{4}^{*}  \rangle, ~~{\rm if }~2m=4.
       \end{array}
    \right.
\end{equation}

Then we obtain
\begin{thm}
The spectral noncommutative residue for the  differential operator $[(D+c(X)\gamma)^2,f]$  multiplied by $c(u) $,$c(v) $ is equal to
 \begin{align}
& Wres\big(c(u)[(D+c(X)\gamma)^2,f]c(v)(D+c(X)\gamma)^{-2m}\big)  \nonumber\\
=&\left\{
       \begin{array}{c}    \int_{M}  2^{3}{\rm  vol}(S^{3})\Big(-\Delta(f)g(u,v)+ 2g(u,\nabla^{S(TM)}_{{\rm grad}f}v)\Big){\rm dvol}_{ M} ~~~~~~~~~~~~~~~~~~~~~~~~~~~~~~~~~~~~~~~~~~~~~~~~ \\
 + \int_{M} 128 {\rm  vol}(S^{3})\langle u^{*}\wedge v^{*}\wedge ({\rm grad}f)^{*}\wedge X^{*},e_{1}^{*}\wedge e_{2}^{*}\wedge e_{3}^{*}\wedge e_{4}^{*}  \rangle   ;  ~~~~~~~~~~~~~~~~~~~~~~~~~~~~ ~~~~~~2m=4; \\[2pt]
     0, ~~~~~~~~~~~~~~~~~~~~~~~~~~~~~~~~~~~~~~ ~~~~~~~~~~~~~~~~~~~~~~~~~~~~~~~~~~  ~~~~  ~~~~~~  ~~~~~~~~~~~~~~~~~~~~ ~2m\neq 4 .
       \end{array}
    \right.
\end{align}
\end{thm}

{\bf Case (3)} The  spectral  divergence theorem for  $D + \sqrt{-1}c(T)\gamma $ by the 3-form $T$ and the grading operator $\gamma$ on $S(TM)$.

Substituting $D + \sqrt{-1}c(T)\gamma $ into (4.16) yields
 \begin{align}
&\sigma_{-2m}(c(u)[(D+\sqrt{-1}c(T)\gamma)^2,f]c(v)(D+\sqrt{-1}c(T)\gamma)^{-2m})(x_0)  \nonumber\\
=&\Big(H_{1}(x_0)+H_{2}(x_0)+H_{3}(x_0)+H_{4}(x_0)\Big)\Big|_{\Phi=\sqrt{-1}c(T)\gamma},
\end{align}
where
\begin{align}
H_{1}(x_0)\Big|_{\Phi=\sqrt{-1}c(T)\gamma}=&\sum_{j=1}^n\partial_{x_j}(f)c(u) \sum_{i=1}^{2m}\big(c(\partial_{i})\sqrt{-1}c(T)\gamma +\sqrt{-1}c(T)\gamma c(\partial_{i})\big)c(v) |\xi|^{-2m};\nonumber\\
  H_{2}(x_0)\Big|_{\Phi=\sqrt{-1}c(T)\gamma}=& -\sum_{jl=1}^n\partial_{x_j}\partial_{x_l}(f)g^{jl}c(u)c(v)|\xi|^{-2m};\nonumber\\
H_{3}(x_0)\Big|_{\Phi=\sqrt{-1}c(T)\gamma}=& -2 \sum_{j=1}^n\partial_{x_j}(f) c(u)  \partial_{x_j}(c(v))|\xi|^{-2m};\nonumber\\
   H_{4}(x_0)\Big|_{\Phi=\sqrt{-1}c(T)\gamma}=&-2m \sum_{j=1}^n\partial_{x_j}(f) g^{jl} \xi_lc(u)c(v)\sum_{i=1}^{2m}\big(c(\partial_{i})\sqrt{-1}c(T)\gamma
   +\sqrt{-1}c(T)\gamma c(\partial_{i})\big)\xi_j  |\xi|^{-2m}.\nonumber
\end{align}

(1) When $n=2m=6$, by  Theorem 4.12 in \cite{WW4} we obtain
\begin{align}
{\rm{Tr}}\big(c(u)c(v)c({\rm grad}f)c(T)\gamma )\big)
=&\sum_{1 \leq  j,k,l\leq n} T(e_{j},e_{k},e_{l})
{\rm{Tr}}\big(c(u)c(v)c({\rm grad}f)c(e_{j})c(e_{k})c(e_{l})\gamma\big)\nonumber\\
=&\sum_{1 \leq  j,k,l\leq n} T(e_{j},e_{k},e_{l})
{\rm{ Str}}\big(c(u)c(v)c({\rm grad}f)c(e_{j})c(e_{k})c(e_{l})\big)\nonumber\\
=&\frac{2^{3}}{(\sqrt{-1})^{3}}\langle  u^{*}\wedge v^{*}\wedge ({\rm grad}f)^{*}\wedge T,
e_{1}^{*}\wedge e_{2}^{*}\wedge e_{3}^{*}\wedge e_{4}^{*}\wedge e_{5}^{*}\wedge e_{6}^{*}  \rangle.
\end{align}

(2) When $n=2m=4$, by  Theorem 4.12 in \cite{WW4} we obtain
\begin{align}
&{\rm{Tr}}\big(c(u)c(v)c({\rm grad}_{f})c(T)\gamma )\big)\nonumber\\
=&\sum_{1 \leq  j,k,l\leq n} T(e_{j},e_{k},e_{l})
{\rm{Tr}}\big(c(u)c(v)c({\rm grad}_{f})c(e_{j})c(e_{k})c(e_{l})\gamma\big)\nonumber\\
=&\sum_{1 \leq  j,k,l\leq n} T(e_{j},e_{k},e_{l})
{\rm{ Str}}\big(c(u)c(v)c({\rm grad}f)c(e_{j})c(e_{k})c(e_{l})\big)\nonumber\\
=&\frac{2^{2}}{(\sqrt{-1})^{2}} \langle -g(u,v)({\rm grad}f)^{*}\wedge T+ g(u,{\rm grad}f)v^{*}\wedge T-g(v,{\rm grad}f)u^{*}\wedge T,
e_{1}^{*}\wedge e_{2}^{*}\wedge e_{3}^{*}\wedge e_{4}^{*}   \rangle.
\end{align}

Then we obtain
\begin{thm}
The spectral noncommutative residue for the  differential operator $[(D+\sqrt{-1}c(T)\gamma)^2,f]$  multiplied by $c(u) $,$c(v) $  read
\begin{align}
& Wres\big(c(u)[(D+\sqrt{-1}c(T)\gamma)^2,f]c(v)(D+\sqrt{-1}c(T)\gamma)^{-2m}\big)  \nonumber\\
=&\left\{
       \begin{array}{c}  \int_{ M}   2^{3}\Big(-\Delta(f)g(u,v)+ 2g(u,\nabla^{S(TM)}_{{\rm grad}f}v)\Big) {\rm{vol}}(S^{3}){\rm dvol}_{ M}~~~~~~~~~~~~~~~~~ ~~~~~~~ ~~~~~~~~~~~~~ ~~~~~~~ ~~~~~~~~~~~~ \\
        ~~~~~~~+ 2^{4}\sqrt{-1}\int_{ M} (-g(u,v)({\rm grad}f)^{*}\wedge T+ g(u,{\rm grad}f)v^{*}\wedge T-g(v,{\rm grad}f)u^{*}\wedge T){\rm{vol}}(S^{3}) , ~~~n=4; \\[2pt]
       \int_M  2^{2}\Big(-\Delta(f)g(u,v)+ 2g(u,\nabla^{S(TM)}_{{\rm grad}f}v)\Big) {\rm{vol}}(S^{5}){\rm dvol}_{ M}
       +2^{5}(u^{*}\wedge v^{*}\wedge ({\rm grad}f)^{*}\wedge T){\rm{vol}}(S^{5}) ,  ~~~n=6; \\[2pt]
     0, ~~~~~~~~~~~~~~~~~~~~~~~~~~~~~~~~~~~~~~~~~~~~~~~~~~~~~~~~~~~~~~~~~~~~~~~~~~~~~~~~~~~~~~~~~~~~~~~~~~~~~~~~~n\neq 4,n\neq 6.
       \end{array}
    \right.
\end{align}
\end{thm}

\subsection{The divergence theorem for compact manifolds with boundary associated with $D +(c(T)+\sqrt{-1}c(Y))$  multiplied by $c(u) $,$c(v) $ }
The purpose of this section is to specify the divergence theorem  on manifold with boundary associated with $D +(c(T)+\sqrt{-1}c(Y))$.

$\mathbf{Case ~~(1)}$:

Let $M$ be a compact oriented Riemannian manifold of even dimension $n=2m$ with boundary $\partial M$,
 \begin{align}
 &\widetilde{{\rm Wres}}[\pi^+(c(u) [(D +c(T)+\sqrt{-1}c(Y))^2,f]c(v)(D +c(T)+\sqrt{-1}c(Y))^{-1} )\nonumber\\
\circ &\pi^+(D +(c(T)+\sqrt{-1}c(Y)))^{-2m }]=L_{1}(x_0)+ L_{2}(x_0),
\end{align}
where
 \begin{align}
L_{1}(x_0)=&\int_M\int_{|\xi|=1}{\rm Tr}_{S(TM)}[\sigma_{-n}(c(u)[(D +(c(T)+\sqrt{-1}c(Y)))^2,f]c(v) \nonumber\\
&\times (D +(c(T)+\sqrt{-1}c(Y)))^{ -2m+1})]\sigma(\xi) {\rm d}x;
\end{align}
  and
  \begin{align}
L_{2}(x_0) =&\int_{|\xi'|=1}\int^{+\infty}_{-\infty}
{\rm Tr}_{S(TM) }
[ \sigma^+_{0}(c(u) [(D +(c(T)+\sqrt{-1}c(Y)))^2,f]c(v)(D +c(T)+\sqrt{-1}c(Y))^{-1}) \nonumber\\
 & \times \partial_{\xi_n}\sigma_{-2m+1}((D +(c(T)+\sqrt{-1}c(Y)))^{-2m+1})(x',0,\xi',\xi_n)]\rm{d}\xi_{2m}\sigma(\xi')\rm{d}x'.
\end{align}

(1) Explicit representation of  $L_{1}(x_0)$

 By Lemma 2.4, Lemma 3.2 and Lemma 4.7, we get the following   main symbol represents.
\begin{lem}The symbols of $[(D +(c(T)+\sqrt{-1}c(Y)))^2,f]$ are given
\begin{align}
\sigma_1((D +(c(T)+\sqrt{-1}c(Y)))^2)=& \sqrt{-1}(\Gamma^j-2\sigma^j)\xi_{j} -\sum_{i=1}^{2m}\big(c(\partial_{i})(c(T)+\sqrt{-1}c(Y))\nonumber\\
  & +(c(T)+\sqrt{-1}c(Y)) c(\partial_{i})\big)\partial_{i};\\
\sigma_2((D +(c(T)+\sqrt{-1}c(Y)))^2)=& |\xi|^{2} ;\\
\sigma_0([(D +(c(T)+\sqrt{-1}c(Y)))^2,f])=&\sum_{j=1}^n\partial_{x_j}(f)\Big(\Gamma^j-2\sigma^j+\sum_{i=1}^{2m}\big(c(\partial_{i})(c(T)+\sqrt{-1}c(Y)) \nonumber\\
& +(c(T)+\sqrt{-1}c(Y)) c(\partial_{i})\big)\Big) -\sum_{jl=1}^n\partial_{x_j}\partial_{x_l}(f)g^{jl};\\
\sigma_1([(D +(c(T)+\sqrt{-1}c(Y)))^2,f])=&-2\sqrt{-1}\sum_{jl=1}^n\partial_{x_j}(f)g^{jl}\xi_l.
\end{align}
\end{lem}
By Lemma 4.15, and substituting $D +(c(T)+\sqrt{-1}c(Y))$ into (4.16) yields
 \begin{align}
&\sigma_{-2m}(c(u)[(D +(c(T)+\sqrt{-1}c(Y)))^2,f]c(v)(D +(c(T)+\sqrt{-1}c(Y)))^{-2m})(x_0)  \nonumber\\
=&\Big(H_{1}(x_0)+H_{2}(x_0)+H_{3}(x_0)+H_{4}(x_0)\Big)\Big|_{\Phi=c(T)+\sqrt{-1}c(Y)},
\end{align}
where
\begin{align}
H_{1}(x_0)\Big|_{\Phi=c(T)+\sqrt{-1}c(Y)}=&\sum_{j=1}^n\partial_{x_j}(f)c(u) \sum_{i=1}^{2m}\big(c(\partial_{i})(c(T)+\sqrt{-1}c(Y)) \nonumber\\ &+(c(T)+\sqrt{-1}c(Y))c(\partial_{i})\big)c(v) |\xi|^{-2m};\nonumber\\
  H_{2}(x_0)\Big|_{\Phi=c(T)+\sqrt{-1}c(Y)}=& -\sum_{jl=1}^n\partial_{x_j}\partial_{x_l}(f)g^{jl}c(u)c(v)|\xi|^{-2m};\nonumber\\
H_{3}(x_0)\Big|_{\Phi=c(T)+\sqrt{-1}c(Y)}=& -2 \sum_{j=1}^n\partial_{x_j}(f) c(u)  \partial_{x_j}(c(v))|\xi|^{-2m};\nonumber\\
   H_{4}(x_0)\Big|_{\Phi=c(T)+\sqrt{-1}c(Y)}=&-2m \sum_{j=1}^n\partial_{x_j}(f) g^{jl} \xi_lc(u)c(v)\sum_{i=1}^{2m}\big(c(\partial_{i})(c(T)+\sqrt{-1}c(Y))\nonumber\\
   &+(c(T)+\sqrt{-1}c(Y)) c(\partial_{i})\big)\xi_j  |\xi|^{-2m}.\nonumber
\end{align}
  By the relation of the Clifford action and Lemma 4.3 in \cite{WW4}, we obtain
 \begin{align}
&{\rm{Tr}}\Big(c(u) \sum_{i=1}^{2m}\big(c(\partial_{i})( c(T)+\sqrt{-1}c(Y))+( c(T)+\sqrt{-1}c(Y)) c(\partial_{i})\big)c(v) \Big)\nonumber\\
=& \sum_{i=1}^{2m}{\rm{Tr}}   \big(c(u) c(\partial_{i})  c(T) c(v) \big)+\sum_{i=1}^{2m}{\rm{Tr}}   \big(c(u) c(T) c(\partial_{i})  c(v) \big)\nonumber\\
  & +\sum_{i=1}^{2m}{\rm{Tr}}   \big(c(u) \big( c(\partial_{i})\sqrt{-1}c(Y) +\sqrt{-1}c(Y) c(\partial_{i})\big)   c(v) \big)\nonumber\\
=&\Big(-2\sum_{i=1}^{2m}T(u,v,\partial_{i})   +2g(Y,\partial_{i})g(u,v) \Big){\rm Tr}_{S(TM)} [ {\rm Id}].
\end{align}
And
 \begin{align}
&{\rm{Tr}}\Big(c(u)c(v) \sum_{i=1}^{2m}\big(c(\partial_{i})( c(T)+\sqrt{-1}c(Y))+( c(T)+\sqrt{-1}c(Y)) c(\partial_{i})\big) \Big)\nonumber\\
=& \sum_{i=1}^{2m}{\rm{Tr}}   \big(c(u)c(v) c(\partial_{i})  c(T) \big)+\sum_{i=1}^{2m}{\rm{Tr}}   \big(c(u) c(v) c(T) c(\partial_{i}) \big)\nonumber\\
  & +\sum_{i=1}^{2m}{\rm{Tr}}   \big(c(u)c(v) \big( c(\partial_{i})\sqrt{-1}c(Y) +\sqrt{-1}c(Y) c(\partial_{i})\big)   \big)\nonumber\\
=&\Big(-2\sum_{i=1}^{2m}T(u,v,\partial_{i})-2g(Y,\partial_{i})g(u,v) \Big){\rm Tr}_{S(TM)} [ {\rm Id}].
\end{align}

Then
 \begin{align}
L_{1}(x_0)=&\int_M\int_{|\xi|=1}{\rm Tr}_{S(TM)}[\sigma_{-n}(c(u)[(D +(c(T)+\sqrt{-1}c(Y)))^2,f]c(v) \nonumber\\
&\times (D +(c(T)+\sqrt{-1}c(Y)))^{ -2m })]\sigma(\xi) {\rm d}x \nonumber\\
=&\int_{ M}  \Big(-\Delta(f)g(u,v)+ 2g(u,\nabla^{S(TM)}_{{\rm grad}f}v)-4 T(u,v,{\rm grad}f)   \Big){\rm Tr}_{S(TM)}[ {\rm Id}]\rm{vol}_{M}.
\end{align}

 (2) Explicit representation of  $L_{2}(x_0)$

An easy calculation gives
  \begin{align}
& \pi^+_{\xi_n}\big( \sigma_{0}( c(u)[(D +c(T)+\sqrt{-1}c(Y))^2,f]c(v)(D +c(T)+\sqrt{-1}c(Y))^{-1})\big)\nonumber\\
=& \pi^+_{\xi_n}\big( \sigma_{1}(  c(u)[(D +c(T)+\sqrt{-1}c(Y))^2,f] c(v)) \sigma_{-1}((D +c(T)+\sqrt{-1}c(Y))^{-1})\big)\nonumber\\
                                  =& \pi^+_{\xi_n}\big( -2\sqrt{-1}\sum_{jl=1}^{2m}\partial_{x_j}(f)g^{jl}\xi_l  c(u) c(v) \frac{\sqrt{-1}c(\xi)}{|\xi|^2}\big)\nonumber\\
                =& \pi^+_{\xi_n}\big( -2\sqrt{-1}\sum_{jl=1}^{2m-1}\partial_{x_j}(f)g^{jl}\xi_l c(u) c(v)\frac{\sqrt{-1}c(\xi)}{|\xi|^2}\big)
                  +\pi^+_{\xi_n}\big( -2\sqrt{-1} \partial_{x_n}(f)c(u) c(v)\xi_n\frac{\sqrt{-1}c(\xi)}{|\xi|^2}\big)\nonumber\\
     =& -2\sqrt{-1}\sum_{jl=1}^{2m-1}\partial_{x_j}(f)g^{jl}c(u) c(v) \xi_l  \frac{c(\xi')+\sqrt{-1}c(\rm{d}x_{n})}{2(\xi_n-i)}
                   -2\sqrt{-1} \partial_{x_n}(f)  c(u) c(v)\frac{\sqrt{-1}c(\xi')-c(\rm{d}x_{n})}{2(\xi_n-i)}.
 \end{align}
Where some basic facts and formulae about Boutet de Monvel's calculus which can be found  in Sec.2 in \cite{Wa1}.
By   Lemma 3.3, we get
  \begin{align}
&\sigma_{-2m+1}((D +(c(T)+\sqrt{-1}c(Y)))^{-2m+1})\nonumber\\
=&\sigma_{-2m+2}((D +(c(T)+\sqrt{-1}c(Y)))^{-2m+2})\sigma_{-1}((D +(c(T)+\sqrt{-1}c(Y)))^{-1})\nonumber\\
=& \frac{\sqrt{-1}c(\xi)}{|\xi|^{2m}},
\end{align}
then
  \begin{align}
&\partial_{\xi_n}\big(\sigma_{-2m+1}((D +(c(T)+\sqrt{-1}c(Y)))^{-2m+1})\big)(x_0)
=\partial_{\xi_n}\big(\frac{\sqrt{-1}c(\xi)}{|\xi|^{2m}}\big)(x_0)\nonumber\\
=&\frac{-2\sqrt{-1}m\xi_nc(\xi')}{(1+\xi_n^2)^{m+1}}+\frac{ \sqrt{-1}(1+\xi_n^{2}-2m\xi_n^{2})c(dx_{n})}{(1+\xi_n^2)^{m+1}}.
\end{align}
By the relation of the Clifford action and $ {\rm{Tr}}(AB)= {\rm{Tr}}(BA) $, we get
\begin{align}
&{\rm Tr}_{S(TM)}\Big[\pi^+_{\xi_n}\big( \sigma_{0}( c(u)[(D +c(T)+\sqrt{-1}c(Y))^2,f]c(v)(D +c(T)+\sqrt{-1}c(Y))^{-1})\big)\nonumber\\
&\times \partial_{\xi_n}\big(\sigma_{-2m+1}((D +(c(T)+\sqrt{-1}c(Y)))^{-2m+1})\big)(x_0)
\Big]\nonumber\\
 =& \sum_{jl=1}^{2m-1}\partial_{x_j}(f)g^{jl}  \xi_l \frac{ -2m \xi_n}{(\xi_n-i)(1+\xi_n^2)^{m+1}}
 {\rm Tr}_{S(TM)}[ c(u)c(v)c(\xi')c(\xi')]\nonumber\\
 &+ \sum_{jl=1}^{2m-1}\partial_{x_j}(f)g^{jl}  \xi_l\frac{ (1+\xi_n^{2}-2m\xi_n^{2})}{(\xi_n-i)(1+\xi_n^2)^{m+1}} {\rm Tr}_{S(TM)} [c(u)c(v)c(\xi')c(dx_{n})]\nonumber\\
   &+  \sum_{jl=1}^{2m-1}\partial_{x_j}(f)g^{jl}  \xi_l\frac{-2\sqrt{-1}m   \xi_n}{(\xi_n-i)(1+\xi_n^2)^{m+1}}
 {\rm Tr}_{S(TM)}[ c(u)c(v)c(dx_{n})c(\xi')]\nonumber\\
 &+  \sum_{jl=1}^{2m-1}\partial_{x_j}(f)g^{jl}  \xi_l\frac{ \sqrt{-1}  (1+\xi_n^{2}-2m\xi_n^{2})}{(\xi_n-i)(1+\xi_n^2)^{m+1}} {\rm Tr}_{S(TM)} [c(u)c(v)c(dx_{n})c(dx_{n})]\nonumber\\
  &+ \frac{-2\sqrt{-1}m \partial_{x_n}(f) \xi_n}{(\xi_n-i)(1+\xi_n^2)^{m+1}}
 {\rm Tr}_{S(TM)}[ c(u)c(v)c(\xi')c(\xi')]\nonumber\\
 &+ \frac{  \sqrt{-1} \partial_{x_n}(f)(1+\xi_n^{2}-2m\xi_n^{2})}{(\xi_n-i)(1+\xi_n^2)^{m+1}} {\rm Tr}_{S(TM)} [c(u)c(v) c(\xi')c(dx_{n}) ]\nonumber\\
  &+ \frac{2 m \partial_{x_n}(f) \xi_n}{(\xi_n-i)(1+\xi_n^2)^{m+1}}
 {\rm Tr}_{S(TM)}[ c(u)c(v)c(dx_{n})  c(\xi')]\nonumber\\
 &- \frac{
  \partial_{x_n}(f)(1+\xi_n^{2}-2m\xi_n^{2})}{(\xi_n-i)(1+\xi_n^2)^{m+1}} {\rm Tr}_{S(TM)} [c(u)c(v)c(dx_{n})c(dx_{n})].
 \end{align}
Also, straightforward computations yield
 \begin{align}
&{\rm{Tr}}\Big(\sum_{j=1}^{2m-1}\partial_{x_j}(f)c(u)c(v)c(\partial_{j})c(dx_{n})  \Big)(x_{0})\nonumber\\
=&  \sum_{j=1}^{2m-1}\partial_{x_j}(f)\Big(g(u,dx_{n})g(v,\partial_{j}) -g(u,\partial_{j})g(v,dx_{n}) +g(u,v)g( dx_{n},\partial_{j})   \Big)(x_{0}){\rm Tr}  [ {\rm Id}]\nonumber\\
=&  \sum_{j=1}^{2m-1}\partial_{x_j}(f)\Big(u_{n}\sum_{j=1}^{2m-1}v_{j} \partial_{j}- v_{n}\sum_{j=1}^{2m-1}u_{j} \partial_{j} \Big)(x_{0}){\rm Tr}  [ {\rm Id}]\nonumber\\
=&  \Big(u_{n} v_{T}(f)- v_{n} u_{T}(f)   \Big)(x_{0}) {\rm Tr}  [ {\rm Id}],
\end{align}
where $ v_{T}(f)(x_{0})=\sum_{j=1}^{2m-1}\partial_{x_j}(f)g(v_{j}, \partial_{j})(x_{0})=\sum_{j=1}^{2m-1}g(v_{j},\partial_{x_j}(f)\partial_{j})(x_{0}) $.
Substituting into above calculation we get
  \begin{align}
L_{2}(x_0) =&\int_{|\xi'|=1}\int^{+\infty}_{-\infty}
{\rm Tr}_{S(TM) }
[ \sigma^+_{0}(c(u) [(D +c(T)+\sqrt{-1}c(Y))^2,f]c(v)(D +c(T)+\sqrt{-1}c(Y))^{-1}) \nonumber\\
 & \times \partial_{\xi_n}\sigma_{-2m+1}((D +(c(T)+\sqrt{-1}c(Y)))^{-2m+1})(x',0,\xi',\xi_n)]\rm{d}\xi_{2m}\sigma(\xi')\rm{d}x'\nonumber\\
 =&\int_{|\xi'|=1}\int^{+\infty}_{-\infty}\Big(
 \sum_{j=1}^{2m-1}\partial_{x_j}(f) \frac{(1+\xi_n^{2}-2m\xi_n^{2})\xi_{j}   }{(\xi_n-i)(1+\xi_n^2)^{m+1}}
 {\rm Tr}_{S(TM)}[ c(u)c(v)c(\xi')c(dx_{n})] \nonumber\\
 &-  \sum_{j=1}^{2m-1}\partial_{x_j}(f)\frac{ 2\sqrt{-1}m \xi_n }{(\xi_n-i)(1+\xi_n^2)^{m+1}} {\rm Tr}_{S(TM)} [c(u)c(v)c(dx_{n})c(\xi')]\nonumber\\
&- \frac{2\sqrt{-1}m \partial_{x_n}(f) \xi_n}{(\xi_n-i)(1+\xi_n^2)^{m+1}}
 {\rm Tr}_{S(TM)}[ c(u)c(v)c(\xi')c(\xi')] \nonumber\\
 &- \frac{   \partial_{x_n}(f)(1+\xi_n^{2}-2m\xi_n^{2})}{(\xi_n-i)(1+\xi_n^2)^{m+1}} {\rm Tr}_{S(TM)} [c(u)c(v)c(dx_{n})c(dx_{n})] \Big)
 \rm{d}\xi_{2m}\sigma(\xi')\rm{d}x'\nonumber\\
 =&  \frac{(2m-2)! 2^{-2m+1}\sqrt{-1}\pi}{(m-1)!m!}   \Big(u_{n} v_{T}(f)- v_{n} u_{T}(f)   \Big) \xi_{j}{\rm Tr}  [ {\rm Id}]                      \nonumber\\
 &+ \frac{(2m-1)! 2^{-2m+1}\sqrt{-1}\pi}{(m-1)!m!} \partial_{x_n}(f) g(u,v){\rm Tr}_{S(TM)} [ {\rm Id}] {\rm{vol}}(S^{2m-2})
{\rm d}\rm{vol}_{\partial_{M}}.
\end{align}
 Summing up (4.43) and (4.49) leads to the divergence theorem for manifold with boundary  as follows.
\begin{thm} Let $M$ be an $n=2m$ dimension compact oriented Riemannian manifold with boundary $\partial M$, the divergence theorem for $[ (D +c(T)+\sqrt{-1}c(Y))^{2},f]$  multiplied by $c(u) $,$c(v) $   read
 \begin{align}
 &\widetilde{{\rm Wres}}[\pi^+(c(u) [(D +(c(T)+\sqrt{-1}c(Y)))^2,f]c(v)(D +c(T)+\sqrt{-1}c(Y))^{-1} )\nonumber\\
\circ &\pi^+(D +(c(T)+\sqrt{-1}c(Y)))^{-2m+1}] \nonumber\\
 =&\int_{ M}  \Big(-\Delta(f)g(u,v)+ 2g(u,\nabla^{S(TM)}_{{\rm grad}f}v)-4 T(u,v,{\rm grad}f)   \Big){\rm Tr}_{S(TM)}[ {\rm Id}]{\rm d}\rm{vol}_{M}\nonumber\\
& +\int_{\partial_{M}} \Big(\frac{(2m-2)! 2^{-2m+1}\sqrt{-1}\pi}{(m-1)!m!}   \Big(u_{n} v_{T}(f)- v_{n} u_{T}(f)   \Big)                  \nonumber\\
 &+ \frac{(2m-1)! 2^{-2m+1}\sqrt{-1}\pi}{(m-1)!m!} \partial_{x_n}(f) g(u,v)\Big){\rm Tr}_{S(TM)} [ {\rm Id}] {\rm{vol}}(S^{2m-2})
{\rm d}\rm{vol}_{\partial_{M}}  .
\end{align}
\end{thm}

$\mathbf{Case ~~II}$:
\begin{align}
&\widetilde{{\rm Wres}}[\pi^+([c(u)(D +(c(T)+\sqrt{-1}c(Y)))^2,f]c(v)(D +(c(T)+\sqrt{-1}c(Y)))^{-2})\nonumber\\
&\circ\pi^+(D +(c(T)+\sqrt{-1}c(Y)))^{-2m+2 }]
=\int_M\int_{|\xi|=1}{\rm Tr}_{S(TM)}[\sigma_{-n}([D^2,f] D^{ -2m+2})]\sigma(\xi)\texttt{d}x+\int_{\partial
M}\widehat{\Psi } ,
\end{align}
where $ r-k-|\alpha|+l-j-1=-2m,~~r\leq -1,l\leq-2m+2$, and $r=-1,~k=|\alpha|=j=0,l=-2m+2$, then
\begin{align}
\widehat{\Psi }=&\int_{|\xi'|=1}\int^{+\infty}_{-\infty}
{\rm Tr}_{S(TM)\otimes F}
[ \sigma^+_{-1}( [c(u)(D +(c(T)+\sqrt{-1}c(Y)))^2,f]c(v)(D +(c(T)+\sqrt{-1}c(Y)))^{-2}) \nonumber\\
& \partial_{\xi_n}\sigma_{-2m+2}((D +(c(T)+\sqrt{-1}c(Y)))^{-2m+2 })(x',0,\xi',\xi_n)]\rm{d}\xi_{2m}\sigma(\xi')\rm{d}x'.
\end{align}
An easy calculation gives
  \begin{align}
& \pi^+_{\xi_n}\big( \sigma_{-1}( c(u)[(D +(c(T)+\sqrt{-1}c(Y)))^2,f]c(v)(D +(c(T)+\sqrt{-1}c(Y)))^{-2})\big)\nonumber\\
=& \pi^+_{\xi_n}\big( \sigma_{1}( c(u)[(D +(c(T)+\sqrt{-1}c(Y)))^2,f])c(v) \sigma_{-2}((D +(c(T)+\sqrt{-1}c(Y)))^{-2})\big)\nonumber\\
                                  =& \pi^+_{\xi_n}\big( -2\sqrt{-1}\sum_{jl=1}^{2m}\partial_{x_j}(f)g^{jl}\xi_l c(u)c(v) \frac{1}{|\xi|^2}\big)\nonumber\\
                =& \pi^+_{\xi_n}\big( -2\sqrt{-1}\sum_{jl=1}^{2m-1}\partial_{x_j}(f)g^{jl}\xi_l c(u)c(v)\frac{1}{|\xi|^2}\big)
                  +\pi^+_{\xi_n}\big( -2\sqrt{-1} \partial_{x_n}(f)\xi_n c(u)c(v)\frac{1}{|\xi|^2}\big)\nonumber\\
     =& -2\sqrt{-1}\sum_{jl=1}^{2m-1}\partial_{x_j}(f)g^{jl}\xi_l  \frac{-\sqrt{-1}}{2(\xi_n-i)}c(u)c(v)
                   -2\sqrt{-1} \partial_{x_n}(f)  \frac{1}{2(\xi_n-i)}c(u)c(v).
 \end{align}
And
   \begin{align}
 \partial_{\xi_n}\big(\sigma_{-2m+2}((D +(c(T)+\sqrt{-1}c(Y)))^{-2m+2})\big)(x_0)
=\partial_{\xi_n}\big(|\xi|^{-2m+2}\big) (x_0)
= \frac{2(1-m)\xi_n }{(1+\xi_n^2)^{m }}.
\end{align}
By the relation of the Clifford action and $ {\rm{Tr}}(AB)= {\rm{Tr}}(BA) $, we get
\begin{align}
& {\rm Tr}_{S(TM)}[\sigma^+_{-1}([c(u)(D+(c(T)+\sqrt{-1}c(Y)))^2,f]c(v)(D +(c(T)+\sqrt{-1}c(Y)))^{-2})  \nonumber\\
\times&\partial_{\xi_n}\sigma_{-2m+2}((D +(c(T)+\sqrt{-1}c(Y)))^{-2m+2})(x',0,\xi',\xi_n)] \nonumber\\
 =& \frac{2(m-1)  \partial_{x_j} \sum _{j=1}^{2m-1}\partial_{x_j}(f) \xi_j}{(\xi_n-i)(1+\xi_n^2)^{m }}
 {\rm Tr}_{S(TM)} [ c(u)c(v)]\nonumber\\
 &+ \frac{ -2\sqrt{-1} \partial_{x_n}(f)(1-m)\xi_n}{(\xi_n-i)(1+\xi_n^2)^{m }} {\rm Tr}_{S(TM)} [ c(u)c(v)].
 \end{align}
Substituting into above calculation we get
\begin{align}
\widehat{\Psi }=&\int_{|\xi'|=1}\int^{+\infty}_{-\infty}
{\rm Tr}_{S(TM)\otimes F}
[ \sigma^+_{-1}( [c(u)(D +(c(T)+\sqrt{-1}c(Y)))^2,f]c(v)(D +(c(T)+\sqrt{-1}c(Y)))^{-2}) \nonumber\\
& \partial_{\xi_n}\sigma_{-2m+2}((D +(c(T)+\sqrt{-1}c(Y)))^{-2m+2 })(x',0,\xi',\xi_n)]\rm{d}\xi_{2m}\sigma(\xi')\rm{d}x'\nonumber\\
=&\int_{|\xi'|=1}\int^{+\infty}_{-\infty}\Big(
\frac{2(m-1)  \partial_{x_j} \sum _{j=1}^{2m-1}\partial_{x_j}(f) \xi_j}{(\xi_n-i)(1+\xi_n^2)^{m }}
 {\rm Tr}_{S(TM)} [ {\rm Id}]\nonumber\\
  &+ \frac{ -2\sqrt{-1} \partial_{x_n}(f)(1-m)\xi_n}{(\xi_n-i)(1+\xi_n^2)^{m }} {\rm Tr}_{S(TM)} [ {\rm Id}]. \Big)
 \rm{d}\xi_{2m}\sigma(\xi')\rm{d}x'\nonumber\\
 =&\frac{- (2m-2)! 2^{-2m+2}\sqrt{-1}\pi}{(m-2)!m!} \partial_{x_n}(f) g(u,v){\rm Tr}_{S(TM)} [ {\rm Id}] {\rm{vol}}(S^{2m-2})
{\rm d}\rm{vol}_{\partial_{M}}.
\end{align}
 Summing up (4.43) and (4.56) leads to the divergence theorem for manifold with boundary  as follows.
\begin{thm} Let $M$ be an $n=2m$ dimension compact oriented Riemannian manifold with boundary $\partial M$, the divergence theorem for $[ (D +c(T)+\sqrt{-1}c(Y))^{2},f]$  multiplied by $c(u) $,$c(v) $   read
 \begin{align}
 &\widetilde{{\rm Wres}}[\pi^+([c(u)(D +(c(T)+\sqrt{-1}c(Y)))^2,f]c(v)(D +(c(T)+\sqrt{-1}c(Y)))^{-2})\nonumber\\
&\circ\pi^+(D +(c(T)+\sqrt{-1}c(Y)))^{-2m+2 }]\nonumber\\
 =&\int_{ M}  \Big(-\Delta(f)g(u,v)+ 2g(u,\nabla^{S(TM)}_{{\rm grad}f}v)-4 T(u,v,{\rm grad}f)   \Big){\rm Tr}_{S(TM)}[ {\rm Id}]{\rm d}\rm{vol}_{M} \nonumber\\
& +\int_{\partial_{M}}  \frac{- (2m-2)! 2^{-2m+2}\sqrt{-1}\pi}{(m-2)!m!} \partial_{x_n}(f) g(u,v)   {\rm Tr}_{S(TM)} [ {\rm Id}] {\rm{vol}}(S^{2m-2})
{\rm d}\rm{vol}_{\partial_{M}}  .
\end{align}
\end{thm}

\section{Conclusions and outlook }

In this paper,  we recover the divergence theorem in the classical differential geometry by using non-commutative residues under the non-commutative geometric framework.
  This approach connects the divergence theorem  and noncommutative geometry, offering a deep insight into the geometric structure of both classical and noncommutative spaces. So it would be interesting to recover other important tensors in both the classical setup as well as for the generalised or quantum geometries. The case of special spectral forms appears particularly worthy of study, in view of the above remarks.  We hope to report on these questions in due course elsewhere.

\section*{ AUTHOR DECLARATIONS}
Conflict of Interest:

The authors have no conflicts to disclose.
\section*{ Author Contributions}
Jian Wang: Investigation (equal); Writing- original draft (equal). Yong Wang: Investigation (equal); Writing-original draft (equal).
\section*{DATA AVAILABILITY}
Data sharing is not applicable to this article as no new data were created or analyzed in this study.

\section*{ Acknowledgements}
The first author was supported by NSFC. 11501414. The second author was supported by NSFC. 11771070.
The authors also thank the referee for his (or her) careful reading and
helpful comments.

\end{document}